\documentclass[reqno,10pt]{article}
\usepackage{amssymb,amsthm,amsmath,amsfonts}
\usepackage{epsfig}

\theoremstyle{plain}
\begingroup
 
\newtheorem{thm}{Theorem}[section] 

\newtheorem{lemma}[thm]{Lemma}

\endgroup
\theoremstyle{definition}
\newtheorem{defn}{Definition}[section]
\newtheorem{rem}[defn]{Remark}

\theoremstyle{remark}

\numberwithin{equation}{section}
\numberwithin{figure}{section}

\begin{document}
\title{Classification of  degree three polynomial solutions to the Polubarinova-Galin equation
}
\author{Yu-Lin Lin}

\date{\today}

\label{firstpage}
\maketitle

\begin{abstract}
The Polubarinova-Galin equation is derived from zero-surface-tension Hele-Shaw flow driven by injection.  In this paper, we classify degree three polynomial solutions to the Polubarinova-Galin equation  into three categories: global solutions, solutions which can be continued after blow-up and solutions which cannot be continued after blow-up. The coefficient region of the initial functions  in each category is obtained. \end{abstract} 
 \section{Introduction}
 We first list some useful notation. For any  set $\mathbb{E}$ which contains the origin, we can define
 \[
 \begin{array}{ll}
 &\mathcal{H}(\mathbb{E})=\left\{f\mid \mbox{$f(\zeta)$ is analytic in a neighborhood of $\mathbb{E}$, $f(0)=0$ and $f^{'}(0)>0$ }\right\}.\\
  &\mathcal{O}_{loca}(\mathbb{E})=\left\{f\mid \mbox{$f(\zeta)\in\mathcal{H}(\mathbb{E}), f^{'}(\zeta)\neq 0$ for $\zeta\in\mathbb{E}$} \right\}.\\
& \mathcal{O}_{univ}(\mathbb{E})=\left\{f\mid \mbox{$f(\zeta)\in\mathcal{O}_{loca}(\mathbb{E})$ is univalent in a neighborhood of $\mathbb{E}$ }\right\}.\\
& \mathcal{P}_{n}=\left\{f|f=\sum_{j=1}^{n}a_{j}(0)\zeta^{j}, a_{1}(0)>0, a_{n}(0)>0\right\}.
 \end{array}\]
   \[\mathcal{P}_{n,univ}(\mathbb{E})=\mathcal{O}_{univ}(\mathbb{E})\cap\mathcal{P}_{n},\quad \mathcal{P}_{n,loca}(\mathbb{E})=\mathcal{O}_{loca}(\mathbb{E})\cap\mathcal{P}_{n}.\]
   We set
  \[\mathbb{D}=\left\{z||z|<1, z\in\mathbb{C}\right\}, \quad f_{t}(\zeta,t)=\frac{\partial}{\partial t}f(\zeta,t), \quad f^{'}(\zeta,t)=\frac{\partial}{\partial \zeta}f(\zeta,t).
 \]
 In this paper, we mainly deal with zero-surface-tension (ZST) Hele-Shaw flow driven by injection at the origin with the strength $2\pi$.  Denote the moving domain to be $\Omega(t), t\geq 0$. In 1948, Polubarinova and Galin consider the case that $\Omega(t)$ is simply-connected and that $\partial\Omega(t)$ is real-analytic. In this case, there exist conformal mappings $f(\cdot,t)\in \mathcal{O}_{univ}(\overline{\mathbb{D}})$ such that $f(\mathbb{D},t)=\Omega(t)$. Polubarinova and Galin show that $f(\zeta,t)$ satisfies 
 \begin{equation}
 \label{pg}
  Re\left[f_{t}(\zeta,t)\overline{f^{'}(\zeta,t)\zeta}\right]=1, \zeta\in\partial \mathbb{D}
  \end{equation}
  and it is called the Polubarinova-Galin equation (the P-G equation). If $f(\zeta,t)\in \mathcal{O}_{univ}(\overline{\mathbb{D}})$ satisfies the P-G equation (\ref{pg})
 and  is continuously differentiable with respect to $t$ in  $[0,\epsilon)$, we call it  by an univalent (or strong) solution to (\ref{pg}).  If the univalent solution $f(\zeta,t)$ ceases to exist as $t=t_{1}$, we say that  the univalent solution $f(\zeta,t)$ blows up at $t=t_{1}$. Similarly, we can also define a locally univalent solution to (\ref{pg}) if $f(\zeta,t)\in \mathcal{O}_{loca}(\overline{\mathbb{D}})$ satisfies the P-G equation (\ref{pg}) and  is continuously differentiable with respect to $t$ in  $[0,\epsilon)$.  If the locally univalent solution $f(\zeta,t)$ ceases to exist as $t=t_{2}$, we say that  this locally univalent solution $f(\zeta,t)$ blows up at $t=t_{2}$. \par
  The short-time well-posedness for univalent solutions  has been proven in Reissig and Wolfersdorf~\cite{reissig}. In Gustafsson~\cite{gustaf1}, the author  shows that if initial functions $f(\zeta,0)\in\mathcal{O}_{univ}(\overline{\mathbb{D}})$ are rational functions of $\zeta$, then the univalent solutions $f(\zeta,t)$ are rational functions of $\zeta$ also, and though the poles of $f(\zeta,t)$ move around in the complex plane, they never collide or produce new poles. Furthermore, if the initial functions  are polynomials of $\zeta$, then the solutions are polynomials of the same degree. These results in Reissig and Wolfersdorf~\cite{reissig} and Gustafsson~\cite{gustaf1}  also hold for  the locally univalent solutions. In this paper, we mainly deal with polynomial solutions. For a given polynomial   $f(\zeta,0)=\sum_{j=1}^{n}a_{j}(0)\zeta^{j}$ which is univalent in $\overline{\mathbb{D}}$, we can assume $a_{1}(0)>0, a_{n}(0)>0$ by suitable normalization and hence $f(\zeta,0)\in\mathcal{P}_{n,univ}(\overline{\mathbb{D}})$.\par
  In this paper, we mainly deal with initial functions in the set $\mathcal{P}_{3,univ}(\overline{\mathbb{D}})$. We  can classify $f(\zeta,0)\in\mathcal{P}_{3,univ}(\overline{\mathbb{D}})$ into three categories:
  \begin{enumerate}
  \item[($C_{1}$)] $f(\zeta,0)\in\mathcal{P}_{3,univ}(\overline{\mathbb{D}})$ which gives rise to a global strong solution to (\ref{pg}).
  \item [($C_{2}$)] $f(\zeta,0)\in\mathcal{P}_{3,univ}(\overline{\mathbb{D}})$ which gives rise to a  strong solution to (\ref{pg}) that  blows up at $t=t^{*}$ but  can be continued after that.
 \item [($C_{3}$)] $f(\zeta,0)\in\mathcal{P}_{3,univ}(\overline{\mathbb{D}})$ which gives rise to a strong solution to (\ref{pg})  that  blows up at $t=t^{*}$ and cannot be continued after that.
  \end{enumerate}
 The main goal of this paper is to give a clear description about the coefficient region of  each category. The coefficient regions of the first two categories are obtained in this paper.  As for the third category, we can always express it  to be the complement of the first two categories in  $\mathcal{P}_{3,univ}(\overline{\mathbb{D}})$. However, the coefficient region of the third category cannot be expressed explicitly since the coefficient region of $\mathcal{P}_{3,univ}(\overline{\mathbb{D}})$ is still a difficult open problem and we are not going to address it  in this paper. This way of classifying solutions has been suggested in  Howison \cite{obstacle}. For these univalent solutions which blow up at finite time but can be continued after blow-up, Howison \cite{obstacle} has many discussions about their singularity  by explaining  the ZST Hele-Shaw problem to be an obstacle problem.\par
   Since we obtain the coefficient region of functions in $\mathcal{P}_{3,univ}(\overline{\mathbb{D}})$ which give rise to global univalent solutions to (\ref{pg}),  we would like to review some past work regarding global dynamics of univalent solutions to (\ref{pg}). In Gustafsson  and Prokhorov and
              Vasil'ev \cite{gustaf2}, it is shown that starlike univalent functions in $\mathcal{O}_{univ}(\overline{\mathbb{D}})$ must give rise to global strong solutions to (\ref{pg}). In Huntingford \cite{hford}, the author considers degree three polynomials $f(\zeta,0)=a_{1}(0)\zeta+a_{2}(0)\zeta^2+a_{3}(0)\zeta^3\in \mathcal{P}_{3,univ}(\overline{\mathbb{D}}), a_{2}(0)\in \mathbb{R}$ and finds the exact coefficient region for  these $f(\zeta,0)$ which give rise to  global strong solutions to (\ref{pg}). In this work, we generalize the result in Huntingford \cite{hford} to the general case that $a_{2}(0)\in\mathbb{C}$.\par
 One of the difficulties to obtain our result  is due to the lack of explicit description for the coefficient region of  functions in $\mathcal{P}_{3,univ}(\overline{\mathbb{D}})$. In Brannan \cite{brannan}, there are some necessary and sufficient conditions for the coefficients of  functions in $\mathcal{P}_{3,univ}(\overline{\mathbb{D}})$. In this former work, the author also determines the coefficient region of real-valued coefficient polynomials in $ \mathcal{P}_{3,univ}(\overline{\mathbb{D}})$. However,  the coefficient region of all functions in $\mathcal{P}_{3,univ}(\overline{\mathbb{D}})$ is not fully determined yet.  Therefore, we need to find some tools to obtain our results. \par
 The main tool to obtain our main result is Theorem 4.1 in Gustafsson and Lin~\cite{gustaf4} which shows that global locally univalent solutions to (\ref{pg}) are global univalent solutions to (\ref{pg}). In this work, we first determine  the coefficient region of  functions in $\mathcal{P}_{3,loca}(\overline{\mathbb{D}})$  and then use Theorem 4.1 in Gustafsson and Lin~\cite{gustaf4} to obtain the coefficient region of  functions in $\mathcal{P}_{3,univ}(\overline{\mathbb{D}})$ which give rise to global univalent solutions to (\ref{pg}) .\par
 The organization of this paper is as follows: In Section~\ref{sec2}, we determine the coefficient region of initial functions of global univalent degree three  polynomial solutions. In Section \ref{sec3}, we classify these functions in $\mathcal{P}_{3,univ}(\overline{\mathbb{D}})$  which fail to give rise to  global strong  solutions  by checking if the strong solutions can be continued after blow-up. In Section \ref{sec4}, we have discussion about the possibility of applying our method to polynomial solutions of degree $n\geq 4$.
  \section{Coefficient region of initial functions of global strong degree three polynomial solutions}
  \label{sec2}
  The following is the list of notations we use in the rest of this paper.
 \begin{defn}
Define          
        \[\mathcal{P}_{n,glob}(\overline{\mathbb{D}})=\left\{f|f\in \mathcal{P}_{n,univ}(\overline{\mathbb{D}})\quad \mbox{gives rise to a global univalent solution to (\ref{pg})}\right\}. \]
        \end{defn}
        \begin{defn}
        \label{def3}
 Define  the map $\Lambda:\mathcal{P}_{n}\rightarrow\mathbb{R}^{2n-3}$ by
        \[\Lambda(f)=\left(Re(\frac{a_{2}}{a_{1}}),Im(\frac{a_{2}}{a_{1}}),\cdots,Re(\frac{a_{n-1}}{a_{1}}),Im(\frac{a_{n-1}}{a_{1}}),\frac{a_{n}}{a_{1}} \right), f(\zeta)=\sum_{j=1}^{n}a_{j}\zeta^{j}\in \mathcal{P}_{n}.\]
       With the mapping $\Lambda$, we define
         \[
 \begin{array}{ll}
 &\Lambda_{n,univ}=\left\{\Lambda(f)|f\in\mathcal{P}_{n,univ}(\overline{\mathbb{D}})\right\}\\
   &\Lambda_{n,loca}=\left\{\Lambda(f)|f\in\mathcal{P}_{n,loca}(\overline{\mathbb{D}})\right\}\\
&\Lambda_{n,glob}=\left\{\Lambda(f)|f\in \mathcal{P}_{n,glob}(\overline{\mathbb{D}})\right\}. \end{array}\]
   \end{defn}
   \par
 These sets $\Lambda_{n,univ}$, $\Lambda_{n,loca}$ and  $\Lambda_{n,glob}$ in Definition~\ref{def3} are  the coefficient regions of  $\mathcal{P}_{n,univ}(\overline{\mathbb{D}})$, $\mathcal{P}_{n,loca}(\overline{\mathbb{D}})$ and $\mathcal{P}_{n,glob}(\overline{\mathbb{D}})$ respectively. 

 In Subsection~\ref{subsec21}, we first determine the set $\Lambda_{3,loca}$. In Subsection~\ref{subsec22}, we express degree three polynomial solutions to (\ref{pg}) in terms of  their leading coefficient and conserved quantities in order to simplify our problem. In Subsection~\ref{subsec23}, we determine  $\Lambda_{3,glob}$ completely.
\subsection{Coefficient region of locally univalent degree three polynomials}
\label{subsec21}
The set $\Lambda_{3,loca}$ is determined as follows:
 \begin{thm}
\begin{equation}
\label{local}
   \Lambda_{3,loca}=\left\{\left(x_1,x_2,x_3\right)\in\mathbb{R}^{3}|\frac{x_1^2}{(1+3x_3)^2}+\frac{x_2^2}{(1-3x_3)^2}<\frac{1}{4}, 0< x_3<\frac{1}{3}\right\}.
  \end{equation}
   \end{thm}
   \begin{proof}
   \begin{enumerate}
   \item[($\subseteq$)]
   Let $f(\zeta)=\zeta+(x_1+ix_2)\zeta^2+x_3\zeta^3\in \mathcal{P}_{3,loca}(\overline{\mathbb{D}})$. Assume $t(a+bi), a-bi$ are two zeros of $f^{'}(\zeta)=0$ where $a^2+b^2>1,  t\geq 1$. The relation between $x_1, x_2, x_3$ and $t, a, b$ is as follows:
    \begin{equation}  \left\{ \begin{array}{ll}
        t(a^2+b^2)&= \frac{1}{3x_3}\\
         ta+a&=\frac{-2x_1}{3x_3}\\
         tb-b&=\frac{-2x_2}{3x_3}\end{array} \right. 
         \end{equation}
   where $a^2+b^2>1$ and $t\geq 1$. Let $a=r\cos\theta, b=r\sin\theta, r>1,\theta\in [0,2\pi)$. Then
         \begin{equation}
         \label{eee}
          \left\{ \begin{array}{ll}
        -2x_1&= 3x_3r(t+1)\cos\theta\\
        -2x_2&=3x_3r(t-1)\sin\theta\\
       3x_3&=\frac{1}{ tr^{2}}\end{array} \right. 
       \end{equation}
       where $ r>1,t\geq 1, \theta\in [0,2\pi).$ In the right-hand side of (\ref{eee}), for any pair of $(t,r)$ where $t\geq 1, r>1, tr^2=\frac{1}{3x_3}$, the two values $3x_3r(t+1), 3x_3r(t-1)$ satisfy that
       \begin{equation}
       \label{eee1}
       3x_3r(t+1)\leq1+3x_3,\quad 3x_3r(t-1)\leq 1-3x_3.
       \end{equation}
       Applying (\ref{eee1}) to (\ref{eee}), we conclude that 
       \begin{equation}
        \label{local00}
         \frac{x_1^2}{(1+3x_3)^2}+\frac{x_2^2}{(1-3x_3)^2}<\frac{1}{4}.
         \end{equation}
         Since $3x_3=\frac{1}{tr^2}$ for some $t\geq 1, r>1$, $0<x_3<\frac{1}{3}$ must hold.
         \item[($\supseteq$)]
         Let  $f(\zeta)=\zeta+(x_{1}+ix_{2})\zeta^2+x_{3}\zeta^3$ where  $(x_1,x_2,x_3)$ satisfies (\ref{local00})  and $0<x_{3}<1/3$.  There exists $0< m<1$ such that
         \begin{equation}
         \label{mmm}
         \frac{x_{1}^2}{(1+3x_{3})^2}+\frac{x_{2}^2}{(1-3x_{3})^2}=\frac{1}{4}m^2.
         \end{equation}
         Since the case $x_{2}=0$ has been a known case as shown in Kuzenetsova and Tkachev \cite{local}, we assume $x_{2}\neq 0$ now.
         Define
         \[g(t)=\left(\frac{x_{1}}{3x_{3} r(t+1)}\right)^2+\left(\frac{x_{2}}{3x_{3}r(t-1)}\right)^2, \quad r^2 t=\frac{1}{3x_{3}},  t> 1,  r> 1.\]
          By using (\ref{mmm}), we can obtain that $g((1, \frac{1}{3x_{3}}))=(\frac{1}{4}m^2,\infty)$ and hence there exists $(r_{0},t_{0})$ satisfying $r_{0}>1, t_{0}>1, r_{0}^2t_{0}=\frac{1}{3x_{3}}$ and $g(t_{0})=\frac{1}{4}$. Therefore, there exists $(\cos\theta,\sin\theta),\theta\in [0,2\pi)$ such that
          \[  \left\{ \begin{array}{ll}
        -2x_{1}&= 3x_{3}r_{0}(t_{0}+1)\cos\theta\\
        -2x_{2}&=3x_{3}r_{0}(t_{0}-1)\sin\theta\end{array} \right. \] 
        where $r_{0}>1, t_{0}>1, r_{0}^2t_{0}=\frac{1}{3x_{3}}$. This implies that  $f^{'}(\zeta)=0$ has two zeros $\zeta_{1}=t_{0}r_{0}(\cos\theta+i\sin\theta)$ and $\zeta_{2}=r_{0}(\cos\theta-i\sin\theta)$ where $r_{0}>1, t_{0}>1,\theta\in [0,2\pi)$. Both zeros $|\zeta_{1}|, |\zeta_{2}|\geq r_{0}>1$ and hence $f(\zeta)\in \mathcal{P}_{3,loca}(\overline{\mathbb{D}})$. \end{enumerate}
         \end{proof}

         \subsection{Expressing degree three polynomial solutions in terms of  their leading coefficients and conserved moments }
         \label{subsec22}
          In 1972, Richardson~\cite{richardson} finds some conserved quantities in ZST Hele-Shaw flow driven by injection with the strength $2\pi$. He shows that for moving  Hele-Shaw cell $\Omega(t)$, the complex moments
        \begin{equation}
        M_{k}(t)=\frac{1}{2\pi}\int_{\Omega(t)}\mathbf{z}^{k}d\mathbf{x}d\mathbf{y},  \mathbf{z}=\mathbf{x}+i\mathbf{y}, k\geq 0
        \end{equation}
        satisfy $M_{k}(t)=M_{k}(0), k\geq 1$ and $M_{0}(t)=2t+M_{0}(0)$. In the case that $\Omega(t)=f(\mathbb{D},t)$ where $f(\zeta,t)=\sum_{j=1}a_{j}(t)\zeta^{j}\in\mathcal{O}_{univ}(\overline{\mathbb{D}})$,  these moments $M_{k}(t), k\geq 0$ can be expressed in terms of $a_{j}(t), j\geq 1$ as follows:
        \begin{equation}
        \label{moment}
       M_{k}(t)=\sum_{j_{1},\cdots,j_{k+1}}j_{1}a_{j_{1}}(t)\cdots a_{j_{k+1}}(t)\overline{a_{j_{1}+\cdots+j_{k+1}}(t)}.
        \end{equation}
       Since this paper  mainly deals with polynomials,  we define the following notation for any $f\in  \mathcal{P}_{n}$ in correspondence to (\ref{moment}).
        \begin{defn}
        \label{momentdef}
        For given $ f=\sum_{j=1}^{n}a_{j}\zeta^{j}\in \mathcal{P}_{n}$, we define
       \begin{equation}
       M_{k}(f)=\sum_{j_{1},\cdots,j_{k+1}}j_{1}a_{j_{1}}\cdots a_{j_{k+1}}\overline{a_{j_{1}+\cdots+j_{k+1}}}. 
       \end{equation}
       \end{defn}
                In the rest of this work, without loss of generality, we consider only initial functions $f(\zeta,0)$ where $f^{'}(0,0)=1$. \par
                Lemma \ref{initial07} is useful throughout this paper.         
         \begin{lemma}
         \label{initial07}
         Given $f\in\mathcal{P}_{3}$ and $f^{'}(0)=1$. Let $M_{1}(f)=p+iq, M_{2}(f)=M_{2}$. Then
          \begin{equation}
         \label{initial0}
         \Lambda(f)=\left(\frac{p}{1+3M_{2}},\frac{-q}{1-3M_{2}},M_{2}\right)
         \end{equation}
         where the mapping  $\Lambda$ is defined as in Definition \ref{def3}.
         \end{lemma}
         \begin{proof}
         Let $f=\sum_{j=1}^{3}a_{j}\zeta^{j}\in \mathcal{P}_{3}$. Then
         \[  \left\{ \begin{array}{ll}
        M_{1}(f)&= a_{1}^2\overline{a_{2}}+3a_{1}a_{2}a_{3}\\
        M_{2}(f)&=a_{1}^3a_{3}.\end{array} \right. \] 
    One can do explicit calculation and obtains that
   \begin{equation}
   \label{eq3}
        a_{2}=\frac{pa_{1}^2}{a_{1}^4+3M_{2}}+i\frac{-qa_{1}^2}{a_{1}^4-3M_{2}},\quad a_{3}=\frac{M_{2}}{a_{1}^3}.
        \end{equation}  
        Since $a_{1}=1$, we have
         \begin{equation}
         \label{initial0}
         \Lambda(f)=\left(\frac{p}{1+3M_{2}},\frac{-q}{1-3M_{2}},M_{2}\right).
         \end{equation}
         \end{proof}
        In order to simplify our problem, we  characterize these functions in $ \mathcal{P}_{3,glob}(\overline{\mathbb{D}})$ in terms of their conserved moments in Theorem \ref{gusta}.
        \begin{thm}
        \label{gusta}
        Let $p+iq=M_{1}(f(\cdot,0)), M_{2}=M_{2}(f(\cdot,0))$,$f(\zeta,0)\in \mathcal{P}_{3}, f^{'}(0,0)=1$. Then the following are equivalent.
        \begin{enumerate}
        \item [(a)] $f(\zeta,0)$ gives rise to a global univalent solution to (\ref{pg}).
        \item [(b)]$
        K=\left\{\left(\frac{p\tau}{\tau^4+3M_{2}}, \frac{-q\tau}{\tau^4-3M_{2}}, \frac{M_{2}}{\tau^4}\right)\mid \tau\geq 1\right\}\subset\Lambda_{3,loca}.$
        \item [(c)]$\frac{p^2\tau^{10}}{(\tau^4+3M_{2})^4}+\frac{q^2\tau^{10}}{(\tau^4-3M_{2})^4}<\frac{1}{4},\quad\forall \tau\geq 1.$
        \end{enumerate}
        \end{thm}
         \begin{proof}\par
     \begin{enumerate}
     \item $((a)\Longrightarrow(b))$ Let $f(\zeta,t)=\sum_{j=1}^{3}a_{j}(t)\zeta^{j}$ be the global univalent solution to (\ref{pg}). Then by (\ref{eq3}), 
     \begin{equation}
        a_{2}(t)=\frac{pa_{1}^2(t)}{a_{1}^4(t)+3M_{2}}+i\frac{-qa_{1}^2(t)}{a_{1}^4(t)-3M_{2}},\quad a_{3}(t)=\frac{M_{2}}{a_{1}^3(t)}.
        \end{equation}
    Since $a_{1}(t):[0,\infty)\rightarrow [1,\infty)$ is an onto, increasing function from Howison and Hohlov~\cite{classification} and $\Lambda(f(\cdot,t))\subset\Lambda_{3,univ}$, we obtain that $K\subset\Lambda_{3,loca}$.
     \item $((b)\Longrightarrow (a))$ Assume the trajectory $K\subset\Lambda_{3,loca}$. Then $F(\zeta,\tau)=\tau\zeta+(\frac{p\tau^2}{\tau^4+3M_{2}}+i\frac{-q\tau^2}{\tau^4-3M_{2}})\zeta^2+\frac{M_{2}}{\tau^3}\zeta^3$ satisfies 
     \begin{equation}
  Re\left[F_{\tau}(\zeta,\tau)\overline{F^{'}(\zeta,\tau)\zeta}\right]=Q(\tau)>0, \zeta\in\partial \mathbb{D}.
  \end{equation}
  There exists $a_{1}(t):[0,\infty)\rightarrow [1,\infty)$ such that $\tau=a_{1}(t)$ and $\frac{d}{dt}a_{1}(t)=\frac{1}{Q(\tau)}$. Let
   $f(\zeta,t)=F(\zeta,\tau)$. Then $f(\zeta,t)$ satisfies (\ref{pg}). This solution  $f(\zeta,t)$ is a global locally univalent solution to (\ref{pg}) and hence a global univalent solution to (\ref{pg}) by Theorem 4.1 in Gustafsson and Lin~\cite{gustaf4}.
   \item  $((b)\Longleftrightarrow (c))$   This follows from Theorem~\ref{local} directly.   \end{enumerate}
     \end{proof}
          
         \subsection{Coefficient region of initial functions of global univalent degree three polynomial solutions}
         \label{subsec23}
         In this subsection, we  determine $\Lambda_{3,glob}$ completely in Theorem~\ref{final}. \par
      Let $g_{1,s}(\tau), g_{2,s}(\tau)$ be
        \begin{equation}
        \label{kk}
        g_{1,s}(\tau)=\frac{(5s+\tau^4)(\tau^4+3s)^5}{64s\tau^{14}},\quad g_{2,s}(\tau)=\frac{(5s-\tau^4)(\tau^4-3s)^5}{64s\tau^{14}}.
        \end{equation}
 In Lemma \ref{ggg11},  we would like to analyze  properties of a curve 
          \begin{equation}
          \label{curve}
      g_{s}(\tau)=\left(\frac{g_{1,s}^{\frac{1}{2}}(\tau)}{1+3s},\frac{g_{2,s}^{\frac{1}{2}}(\tau)}{1-3s}\right),\quad 1\leq \tau\leq (5s)^{\frac{1}{4}}, \frac{1}{5}< s<\frac{1}{3}.
       \end{equation}
       A part of the set $\{(g_{s}(\tau),s)|1\leq \tau\leq (5s)^{\frac{1}{4}}, 1/5< s<1/3\}$ represents  the curve  $\partial\Lambda_{3,glob}\cap\Lambda_{3,loca}$  in the first octant  in Theorem \ref{final}.  \par
         \begin{lemma}
        \label{ggg11}
           Let $\tau^{*}(s)=( \max\{1,\sqrt{21}s\})^{\frac{1}{4}}, 1/5< s<1/3$. The following are true.      
        \begin{enumerate}
       \item[(a)] For $1/5<s\leq 1/\sqrt{21}$, 
       \[\left\{(g_{s}(\tau),s)|\tau\in(1, (5s)^{\frac{1}{4}}]\right\}\subset \Lambda_{3,loca}.\]
       For $1/\sqrt{21}<s<1/3$, there exists $\tau^{**}(s)\in(\tau^{*}(s),1/3)$ such that 
       \begin{equation}
       \label{iso}
       \left\{(g_{s}(\tau),s)|\tau\in(1, \tau^{**}(s))\right\}\subset(\Lambda_{3,loca})^{c},\left\{(g_{s}(\tau),s)|\tau\in(\tau^{**}(s), (5s)^{\frac{1}{4}}]\right\}\subset\Lambda_{3,loca}
        \end{equation}
        and $(g_{s}(\tau^{**}(s)),s)\in\partial\Lambda_{3,loca}$.
       \item[(b)]  For $1/5<s<1/3$,\\
        \[ \left\{(g_{s}(\tau),s)|\tau\in[1, (5s)^{\frac{1}{4}}]\right\}\cap\Lambda_{3,loca}=\left\{(g_{s}(\tau),s)|\tau\in[\tau^{*}(s),(5s)^{\frac{1}{4}}]\right\}\cap\Lambda_{3,loca}.\]
            \item[$(c)$] For $1/5<s<1/3$, $\left\{g_{s}(\tau)|\tau\in [\tau^{*}(s),(5s)^{\frac{1}{4}}]\right\}$ is a monotonic curve.
       \item [$(d)$] For $1/5<s<1/3$, $g_{s}(\tau_{1})\neq g_{s}(\tau_{2})$ for $\tau_{1}\neq\tau_{2}$ and $\tau_{1},\tau_{2}\in[\tau^{*}(s),(5s)^{\frac{1}{4}}]$.
       \       \end{enumerate}
         \end{lemma}
           \begin{proof}
    The following three facts are going to be used in the rest of the proof.
    \begin{enumerate}
    \item By Theorem \ref{local}
      \begin{equation}
      \label{eqq1}
    (g_{s}(\tau),s)\in\Lambda_{3,loca}\Longleftrightarrow \left[\frac{g_{1,s}(\tau)}{(1+3s)^4}+\frac{g_{2,s}(\tau)}{(1-3s)^4}\right] <\frac{1}{4}, \frac{1}{5}<s<\frac{1}{3}.
    \end{equation}
    \item
       \begin{equation}
       \label{g1}
       \frac{d}{d\tau}\left[g_{1,s}(\tau)\right]=\frac{5}{32s}\frac{\left[\tau^4+3s\right]^4\left[\tau^8-21s^2\right]}{\tau^{15}},
       \end{equation}
       \begin{equation}
       \label{g2}
       \frac{d}{d\tau}\left[g_{2,s}(\tau)\right]=-\frac{5}{32s}\frac{\left[\tau^4-3s\right]^4\left[\tau^8-21s^2\right]}{\tau^{15}}.
       \end{equation}
       \item
       \begin{equation}
       \label{qqq}
      \begin{array}{ll}
        & \frac{d}{d\tau}\left[\frac{g_{1,s}(\tau)}{(1+3s)^4}+\frac{g_{2,s}(\tau)}{(1-3s)^4}\right] \\
         =&\frac{5(\tau^8-21s^2)}{32s\tau^{15}}\left\{\frac{(\tau^4+3s)^4(1-3s)^4-(\tau^4-3s)^4(1+3s)^4}{(1+3s)^4(1-3s)^4}\right\}\\
           =&\left\{ \begin{array}{ll}
         \leq 0 & \mbox{if $\tau\geq (\sqrt{21}s)^{\frac{1}{4}}$};\\
        \geq 0& \mbox{if $\tau\leq (\sqrt{21}s)^{\frac{1}{4}}$}.\end{array} \right.
        \end{array}
        \end{equation}
        \end{enumerate}
        We separate into two cases $1/5<s\leq 1/\sqrt{21}$ and $1/\sqrt{21}<s<1/3$.\par
        Assume $1/5<s\leq1/\sqrt{21}$ now. By the fact that $[\frac{g_{1,s}(\tau)}{(1+3s)^4}+\frac{g_{2,s}(\tau)}{(1-3s)^4}] =1/4$ as $\tau=1$ and (\ref{qqq}),
        we have that  $[\frac{g_{1,s}(\tau)}{(1+3s)^4}+\frac{g_{2,s}(\tau)}{(1-3s)^4}] <1/4, 1<\tau\leq(5s)^{1/4}$ and hence $(g_{s}(\tau),s)\in\Lambda_{3,loca}$ for $1<\tau\leq (5s)^{1/4}$. From (\ref{g1}) and (\ref{g2}), we can conclude  the monotonicity of the curve $\{g_{s}(\tau)|1\leq\tau\leq(5s)^{1/4}\}$ and that this trajectory $g_{s}(\tau), 1\leq\tau\leq(5s)^{1/4}$ never self-intersects . Hence $(a)-(d)$ hold for $1/5<s\leq 1/\sqrt{21}$.\par
        Assume $1/\sqrt{21}<s<1/3$. By the fact that $[\frac{g_{1,s}(\tau)}{(1+3s)^4}+\frac{g_{2,s}(\tau)}{(1-3s)^4}] =1/4$ as $\tau=1$ and (\ref{qqq}),
        we have that  $[\frac{g_{1,s}(\tau)}{(1+3s)^4}+\frac{g_{2,s}(\tau)}{(1-3s)^4}] >1/4, 1<\tau\leq(\sqrt{21}s)^{1/4}$ and hence $(g_{s}(\tau),s)\notin\Lambda_{3,loca}$ for $1\leq\tau\leq(\sqrt{21}s)^{1/4}$. Since $[\frac{g_{1,s}(\tau)}{(1+3s)^4}+\frac{g_{2,s}(\tau)}{(1-3s)^4}] <1/4$ as $\tau=(5s)^{1/4}$, there exists $\tau^{**}(s)\in ((\sqrt{21}s)^{1/4},(5s)^{1/4})$ such that (\ref{iso}) holds. From (\ref{g1}) and (\ref{g2}), we can conclude the monotonicity of the curve $\{g_{s}(\tau)| (\sqrt{21}s)^{1/4}\leq\tau\leq(5s)^{1/4}\}$ and that this trajectory $g_{s}(\tau), (\sqrt{21}s)^{1/4}\leq\tau\leq(5s)^{1/4}$ never self-intersects. Hence $(a)-(d)$ hold for $1/\sqrt{21}<s<1/3$.\par
        By the above facts for both cases  $1/5<s\leq 1/\sqrt{21}$ and $1/\sqrt{21}<s<1/3$, we can conclude $(a)-(d)$ hold.
        
             \end{proof}
             
             \par
             With Lemma \ref{main}, we would be able to determine $\Lambda_{3,glob}$ completely after obtaining  $\partial\Lambda_{3,glob}$.
                 \begin{lemma}
        \label{main}
        Assume $f(\zeta,0)=\zeta+a_{2}(0)\zeta^2+a_{3}(0)\zeta^3\in\mathcal{P}_{3,loca}(\overline{\mathbb{D}})$ fails to give rise to a global univalent solution to (\ref{pg}). For  $b_{2}(0), b_{3}(0)$ satisfying $|Re(b_{2}(0))|\geq |Re(a_{2}(0))|$, $|Im(b_{2}(0))|\geq|Im(a_{2}(0))|$ and $b_{3}(0)=a_{3}(0)$,   $g(\zeta,0)=\zeta+b_{2}(0)\zeta^2+b_{3}(0)\zeta^3$ also fails to give rise to a global univalent solution to (\ref{pg}).        \end{lemma}
        \begin{proof}
        Let $p+iq=M_{1}(f(\cdot,0)), p_{0}+iq_{0}=M_{1}(g(\cdot,0))$ and $M_{2}=M_{2}(f(\cdot,0))=M_{2}(g(\cdot,0))$. Here $p=(1+3M_{2})Re(a_{2}(0)), q=-(1-3M_{2})Im(a_{2}(0))$ and $p_{0}=(1+3M_{2})Re(b_{2}(0)), q_{0}=-(1-3M_{2})Im(b_{2}(0))$ according to (\ref{initial0}). Since $|Re(b_{2}(0))|\geq |Re(a_{2}(0))|$, $|Im(b_{2}(0))|\geq|Im(a_{2}(0))|$,  we obtain that
    \begin{equation}    
    \label{kk0}
        |p_{0}|=\left|\frac{Re(b_{2}(0))}{1+3M_{2}}\right|\geq \left|\frac{Re(a_{2}(0))}{1+3M_{2}}\right|= |p|, |q_{0}|=\left|\frac{Im(b_{2}(0))}{1-3M_{2}}\right|\geq \left|\frac{Im(a_{2}(0))}{1-3M_{2}}\right|= |q|.
        \end{equation}
        Since $f(\zeta,0)$ fails to give rise to a global univalent solution, by Theorem~\ref{gusta}, 
        \begin{equation}
        \label{kk1}
           \frac{p^2\tau^{10}}{(\tau^4+3M_{2})^4}+\frac{q^2\tau^{10}}{(\tau^4-3M_{2})^4}\geq\frac{1}{4}\quad\mbox{for some $\tau\geq 1$}.
           \end{equation}
           Applying (\ref{kk0}) to (\ref{kk1}), we have
           \begin{equation}
           \label{kk2}
           \frac{p_{0}^2\tau^{10}}{(\tau^4+3M_{2})^4}+\frac{q_{0}^2\tau^{10}}{(\tau^4-3M_{2})^4}\geq\frac{1}{4}\quad\mbox{for some $\tau\geq 1$}
           \end{equation}
           and hence $g(\zeta,0)$ fails to give rise to a  global univalent solution by Theorem ~\ref{gusta}.
  \end{proof}

            The coefficient region  $\Lambda_{3,glob}$ is completely determined as follows:
       \begin{thm}
       \label{final}
        Define a set $A$ by
       \begin{equation}
      \left\{\left(\frac{p}{1+3s},\frac{-q}{1-3s}, s\right)\in\mathbb{R}^{3}|p^2\geq g_{1,s}(\tau), q^2= g_{2,s}(\tau),\tau^{*}(s)\leq \tau\leq (5s)^{\frac{1}{4}}, \frac{1}{5}<s<\frac{1}{3}\right\}
      \end{equation}
      where $g_{1,s}(\tau)$ and $g_{2,s}(\tau)$ are  as defined  in (\ref{kk}) and  $\tau^{*}(s)=( \max\{1,\frac{1}{\sqrt{21}}s\})^{\frac{1}{4}}$.  Then
      \begin{equation}
        \label{conclusion}
      \Lambda_{3,glob}=\Lambda_{3,loca}\cap (A)^{c}.
      \end{equation}
  \end{thm}
  \begin{proof}
  Given $f(\zeta,0)$ satisfying $\Lambda(f(\cdot,0))\in\partial\Lambda_{3,glob}\cap\Lambda_{3,loca}$. Let $M_{1}(f(\cdot,0))=p+iq, M_{2}(f(\cdot,0))=M_{2}$.  Here $0<M_{2}<1/3$ since $\Lambda_{3,glob}\subset\Lambda_{3,loca}$. By Theorem~\ref{gusta}, there exists $\tau> 1$ such that 
 \begin{equation}\label{eq} \left\{ \begin{array}{ll}
         \frac{p^2\tau^{10}}{(\tau^4+3M_{2})^4}+\frac{q^2\tau^{10}}{(\tau^{4}-3M_{2})^4}&=\frac{1}{4}\\
        \frac{\partial}{\partial \tau}\left[\frac{p^2\tau^{10}}{(\tau^4+3M_{2})^4}+\frac{q^2\tau^{10}}{(\tau^4-3M_{2})^4}\right]&=0.\end{array} \right. 
        \end{equation}
        Denote 
        \[r_{1}= \frac{\tau^{10}}{(\tau^4+3M_{2})^4}, \quad r_{2}=\frac{\tau^{10}}{(\tau^{4}-3M_{2})^4}, \quad \alpha_{1}=  \frac{\partial}{\partial \tau}\left[r_{1}\right], \quad \alpha_{2}=\frac{\partial}{\partial \tau}\left[r_{2}\right].\]
              Equation (\ref{eq}) becomes
        \begin{equation}
       \label{eq1} 
        \left\{ \begin{array}{ll}
         p^2r_{1}+q^2r_{2}&=\frac{1}{4}\\
       p^2\alpha_{1}+q^2\alpha_{2}&=0.\end{array} \right. 
       \end{equation}
       Solving (\ref{eq1}), we obtain
       \begin{equation}
       \label{eq2}
       p^2=\frac{-1}{4}\frac{\alpha_{2}}{\alpha_{1} r_{2}-\alpha_{2} r_{1}},\quad q^2=\frac{1}{4}\frac{\alpha_{1}}{\alpha_{1} r_{2}-\alpha_{2} r_{1}}.
       \end{equation}
       Here 
       \begin{equation}
      \alpha_{1}=  \frac{\partial}{\partial \tau}\left[r_{1}\right]=\frac{6\tau^9\left[5M_{2}-\tau^4\right]}{(\tau^4+3M_{2})^5},\quad   \alpha_{2}=\frac{\partial}{\partial \tau}\left[r_{2}\right]=\frac{6\tau^9\left[-5M_{2}-\tau^4\right]}{(\tau^4-3M_{2})^5}        \end{equation}
      and 
      \[  \begin{array}{ll}
         \alpha_{1} r_{2}-\alpha_{2} r_{1}&=\frac{6\tau^{9}(5M_{2}-\tau^{4})}{(\tau^4+3M_{2})^5}\frac{\tau^{10}}{(\tau^4-3M_{2})^{4}}-\frac{6\tau^{9}(-5M_{2}-\tau^4)}{(\tau^4-3M_{2})^5}\frac{\tau^{10}}{(\tau^4+3M_{2})^{4}}\\
        &=\frac{6\tau^{19}}{(\tau^{4}+3M_{2})^{5}(\tau^4-3M_{2})^{5}}\left[(5M_{2}-\tau^4)(\tau^4-3M_{2})-(-5M_{2}-\tau^4)(\tau^4+3M_{2})\right]\\
        &=\frac{6\tau^{19}}{(\tau^4+3M_{2})^5(\tau^4-3M_{2})^5}\left[-\tau^8+8M_{2}\tau^4-15M_{2}^2+\tau^8+8M_{2}\tau^4+15M_{2}^2\right]\\
        &=\frac{96M_{2}\tau^{23}}{(\tau^4+3M_{2})^5(\tau^4-3M_{2})^5}.\end{array} \] 
        Therefore (\ref{eq2}) becomes
        \begin{equation}
        \label{par}
        p^2=g_{1,s}(\tau),\quad q^2=g_{2,s}(\tau),  \quad\mbox{where $s=M_{2}$ and $0<s<\frac{1}{3}$} .      \end{equation}
         \begin{enumerate}
        \item[($a$)] Since there exists $\tau>1$ such that  $g_{1,s}(\tau)\geq 0, g_{2,s}(\tau)\geq 0$, the value $s$ shall satisfy $1/5<s<1/3$.         \item[($b$)]Let $1/5<s<1/3.$ By (\ref{initial0}) and (\ref{par}), $\Lambda_{3,glob}$ is symmetric with respect to $x-$axis and $y-$axis and 
         \begin{equation}
        \partial\Lambda_{3,glob}\cap\Lambda_{3,loca}\cap\{(x_1,x_2,x_3)|x_1\geq 0, x_2\geq 0, x_3=s\}\subseteq\left\{(g_{s}(\tau),s)|1\leq\tau\leq (5s)^{\frac{1}{4}}\right\}\cap\Lambda_{3,loca}
        \end{equation} 
        where $g_{s}(\tau)$ satisfies all properties listed in Lemma \ref{ggg11}. 
         \end{enumerate}
         By (a)-(b) and Lemma \ref{main}, we can describe $\Lambda_{3,glob}$ as (\ref{conclusion}).
  \end{proof}
  \section{Classification of blow-up}
  \label{sec3}
  In Section \ref{sec2}, we completely determine the coefficient region $\Lambda_{3,glob}$. In this Section, we want to completely  classify these functions in $\mathcal{P}_{3,univ}(\overline{\mathbb{D}})$ which give rise to strong solutions blowing up at finite time by checking if these strong solutions can be continued after their blow-up time. Theorem \ref{mainthm} states these classification. \par
  For any point on $\Lambda_{3,univ}\cap (\overline{\Lambda_{3,glob}})^{c}$, the following property holds.
  \begin{lemma}
  \label{global}
Given $f(\zeta,0)\in\mathcal{P}_{3,glob}(\overline{\mathbb{D}})$ satisfying $\Lambda(f(\cdot,0))\in\Lambda_{3,univ}\cap (\overline{\Lambda_{3,glob}})^{c}$. Then $f(\zeta,0)$ gives rise to a strong solution to (\ref{pg}) $f(\zeta,t)$ which  blows up at finite time $t=t^{*}>0$ and cannot be continued after that .
  \end{lemma}
  \begin{proof}
 Without loss of generality, we assume $f(\zeta,0)=\zeta+a_{2}(0)\zeta^2+a_{3}(0)\zeta^{3}$ and $Re(a_{2}(0)), Im(a_{2}(0))$ are nonnegative, $a_{3}(0)>0$. Let $M_{2}=M_{2}(f(\cdot,0))$. If the univalent solution $f(\zeta,t)$ can be continued after $t=t^{*}$ for at least a short time, then 
 \begin{equation}
 \Lambda(f(\cdot,0))=(g_{s}(a_{1}(t^{*})),s),\quad\mbox{where $s=M_{2}\in (0,\frac{1}{3})$ and $a_{1}(t^{*})\in [1,(5s)^{\frac{1}{4}}] $}
 \end{equation}
 and hence $\Lambda(f(\cdot,0))\in\partial \Lambda_{3,glob}$ or $\Lambda(f(\cdot,0))\in(\Lambda_{3,loca})^{c}$ by Lemma \ref{ggg11}. This contradicts to the original assumption for $f(\zeta,0)$.
 \end{proof}
  
For any point on $\partial\Lambda_{3,glob}\cap\Lambda_{3,loca}$, the following property holds.
  \begin{lemma}
  \label{blow}
  Given $f(\zeta,0)\in\mathcal{P}_{3,glob}(\overline{\mathbb{D}})$ satisfying $\Lambda(f(\cdot,0))\in\partial\Lambda_{3,glob}\cap\Lambda_{3,loca}$. Then $f(\zeta,0)$  gives rise to a strong solution to (\ref{pg}) $f(\zeta,t)$ which  blows up at finite time $t=t^{*}>0$ but can be continued after that. As blow-up happens, $f^{'}(\zeta,t^{*})=0$ for some $|\zeta|=1$ and $f(\zeta,t^{*})$ forms cusps of  $5/2$ or $9/2$ types. Moreover, $f(\zeta,t), t>t^{*}$ can be continued as a global univalent solution to (\ref{pg}). 
  \end{lemma}
  \begin{proof}
  Without loss of generality, we assume $f(\zeta,0)=\zeta+a_{2}(0)\zeta^2+a_{3}(0)\zeta^3$. Denote $M_{1}(f(\cdot,0))=p+iq$, $M_{2}(f(\cdot,0))=M_{2}$.  Then by (\ref{initial0}), $a_{2}(0)=\frac{p}{1+3M_{2}}+i\frac{-q}{1-3M_{2}}$ and $a_{3}(0)=M_{2}$. Since $\Lambda_{3,glob}$ is symmetric with respect to $x=0$ and $y=0$ by Theorem~\ref{final}, we can assume $p\geq 0, -q\geq 0$.  \par
  We would like first to show that  $K$ defined in Theorem \ref{gusta} only touches $\partial\Lambda_{3,loca}$ once. Let $g_{s}(\tau)$ ,  $\tau^{*}(s)$ be defined as  in Lemma \ref{ggg11} and $s=M_{2}$. If $K$ touches $\partial\Lambda_{3,loca}$ at least twice,   there exist two distinct real numbers $\tau_{1}, \tau_{2}\in [\tau^{*}(s), (5s)^{1/4}]$ such that   $g_{s}(\tau_{1})=g_{s}(\tau_{2})$ and this contradicts to Lemma \ref{ggg11}.\par
   Let  $K$ touch $\partial\Lambda_{3,loca}$ as $\tau=\tau_{0}^{*}$.  Define 
  \[K_{1}=\left\{\left(\frac{p\tau}{\tau^4+3M_{2}}, \frac{-q\tau}{\tau^4-3M_{2}}\frac{M_{2}}{\tau^4}\right)\mid 1\leq \tau<\tau_{0}^{*}\right\}, \]
  \[K_{2}=\left\{\left(\frac{p\tau}{\tau^4+3M_{2}}, \frac{-q\tau}{\tau^4-3M_{2}}, \frac{M_{2}}{\tau^4}\right)\mid \tau^{*}_{0}< \tau\right\}.\]
  Here, 
   \begin{equation}
   \label{jjj}
   K_{1}, K_{2}\subset\overline{\Lambda_{3,univ}}.
   \end{equation}
 Since $K_{1}\subset\Lambda_{3,loca}$, there exists $t^{*}$ such that  $f(\zeta,0)$ at least gives rise to a locally univalent solution $f(\zeta,t)$ for $0\leq t<t^{*}$ and this locally univalent solution loses locally univalency as $t=t^{*}$.   
 We still need to show that $f(\zeta,t), 0\leq t<t^{*}$ is an univalent solution to (\ref{pg}). If there exists $t^{**}< t^{*}$ such that $f(\zeta,t^{**})$ blows up due to the overlapping of the boundary, then there exists $\epsilon>0$ such that $\Lambda(f(\cdot,t))\in\Lambda_{3,loca}\cap(\overline{\Lambda_{3,univ}})^{c}, t^{**}<t<t^{**}+\epsilon$ and hence $K_{1}\cap(\overline{\Lambda_{3,univ}})^{c}\neq\emptyset$. However, this  contradicts to (\ref{jjj}) and we hence conclude that $f(\zeta,t),0\leq t<t^{*}$ is an univalent solution to (\ref{pg}).  Since $K_{2}\subset\Lambda_{3,loca}$, by Theorem \ref{gusta}, $f(\zeta,t)$ is an univalent solution  for  $t>t^{*}$.\par
Now we assume $f(\zeta,t)$ blows up due to the formation of $m/2$ cusps. From Theorem 2.1 in Gustafsson~\cite{gusta01}, we can conclude that 
 \begin{equation}
 \label{p1}
 \frac{m-1}{2}\leq 4.
 \end{equation}
 By Howison~\cite{obstacle}, if $f(\zeta,t)$ can be continued after blow-up, then 
\begin{equation}
\label{p4}
m\equiv 1\quad(\mbox{mod $4$}).
\end{equation}
By (\ref{p1}) and (\ref{p4}), we conclude that $m=5$ or $9$.

 \end{proof}
  \begin{thm}
  \label{mainthm}
  \begin{enumerate}
  \item[(a)] $f(\zeta,0)\in\mathcal{P}_{3}$, $\Lambda(f(\cdot,0))\in\Lambda_{3,univ}\cap (\overline{\Lambda_{3,glob}})^{c}$ if and only if the  strong solution to (\ref{pg}) $f(\zeta,t)$ blows up at finite time $t=t^{*}$ and cannot be continued after that .
  \item [(b)] $f(\zeta,0)\in\mathcal{P}_{3}$, $\Lambda(f(\cdot,0))\in\partial\Lambda_{3,glob}\cap\Lambda_{3,loca}$ if and only if the  strong solution to (\ref{pg}) $f(\zeta,t)$  blows up at finite time $t=t^{*}>0$ but can be continued after that. As blow-up happens, $f^{'}(\zeta,t^{*})=0$ for some $|\zeta|=1$ and $f(\zeta,t^{*})$ forms cusps of  $5/2$ or $9/2$ types. Moreover, $f(\zeta,t), t>t^{*}$ can be continued as a global univalent solution to (\ref{pg}). \par
  \end{enumerate}
  \end{thm}
 \begin{proof}
By Lemma \ref{global} and Lemma \ref{blow} and the fact that 
\[\Lambda_{3,univ}=\Lambda_{3,glob}\cup \left(\Lambda_{3,univ}\cap (\overline{\Lambda_{3,glob}})^{c}\right)\cup\left( \partial\Lambda_{3,glob}\cap\Lambda_{3,loca}\right),\]
 Theorem \ref{mainthm} must hold.
 \end{proof}
  \begin{rem}
   Huntingford \cite{hford} has shown Theorem \ref{blow} in the case that $f(\zeta,0)$ has real-coefficients.
  \end{rem}
  \section{Discussion }
  \label{sec4}
  Our methods can be also applied to degree $n\geq 4$ univalent polynomial solutions to (\ref{pg}). One can start from the case that $f(\zeta,0)\in\mathcal{P}_{n,univ}(\overline{\mathbb{D}})$ has real-coefficients since  it is known that real-coefficient $f(\zeta,0)\in\mathcal{P}_{n,loca}(\overline{\mathbb{D}})$ gives rise to a real-coefficient univalent solution $f(\zeta,t)\in\mathcal{P}_{n,loca}(\overline{\mathbb{D}})$ to (\ref{pg}) and  there are some known results which can be useful for handling this case. For example,  in Kuznetsova and Tkachev \cite{local}, the authors obtain the coefficient region of real-coefficient polynomials  in $\mathcal{P}_{n,loca}(\overline{\mathbb{D}}),n\geq 3$.  However, one of the difficulties for dealing with the case $n\geq 4$ is that the  calculation can not be as simple as the calculation in the case that $n=3$. Our future work is  to handle the case that $n\geq 4$.

\pagebreak


\bibliography{main0}

\end{document}